\documentclass[conference,letterpaper]{IEEEtran}
\usepackage[short]{optidef}

\usepackage{xcolor}
\usepackage{enumitem}
\usepackage{amsfonts}                                                   \usepackage{balance}
       
\usepackage{epsfig}
\usepackage{amsthm}
\usepackage{graphicx}        
\usepackage{latexsym}
\usepackage{amssymb}
\usepackage{amsmath}
\usepackage{multirow}
\usepackage{cite}
\usepackage{mathtools}
\usepackage{epstopdf}
\usepackage{etoolbox}
\usepackage{array}
\usepackage{setspace}
\usepackage{soul}
\usepackage{mathptmx}
\usepackage[bottom]{footmisc}
\usepackage{tikz}

\usetikzlibrary{decorations.pathreplacing}
\DeclareMathOperator{\trace}{trace}
\newcommand{\diag}{\mathrm{diag}}

\newtheorem{Algorithm}{Algorithm}

\usepackage{verbatim}
\usepackage{calrsfs}
\usepackage{booktabs}
\newcommand\numeq[1]%
  {\stackrel{\scriptscriptstyle(\mkern-1.5mu#1\mkern-1.5mu)}{=}}
\newcommand\geab[1]%
  {\stackrel{\scriptscriptstyle(\mkern-1.5mu#1\mkern-1.5mu)}{\ge}}
\newcommand{\ES}[1]{\textcolor{black}{{#1}}}

\usepackage[titlenumbered,ruled]{algorithm2e}
\DeclareMathAlphabet{\pazocal}{OMS}{zplm}{m}{n}

\let\bbordermatrix\bordermatrix
\patchcmd{\bbordermatrix}{8.75}{4.75}{}{}

\newcommand{\T}{^{\mbox{\tiny T}}}

\newcommand{\rar}{\rightarrow}

\newcommand{\tri}{\triangleq}

\newcommand{\be}{\begin{equation}}
\newcommand{\ee}{\end{equation}}
\newcommand{\bea}{\begin{eqnarray}}
\newcommand{\eea}{\end{eqnarray}}
\newcommand{\bes}{\begin{eqnarray*}}
\newcommand{\ees}{\end{eqnarray*}}
\newcommand{\bce}{\begin{center}}
\newcommand{\ece}{\end{center}}
\newcommand{\beae}{\begin{IEEEeqnarray}{rCl}}
\newcommand{\eeae}{\end{IEEEeqnarray}}

\def\VR{\kern-\arraycolsep\strut\vrule &\kern-\arraycolsep}
\def\vr{\kern-\arraycolsep & \kern-\arraycolsep}

\newcommand{\ben}{\begin{enumerate}}
\newcommand{\een}{\end{enumerate}}

\newcommand{\hso}{\hspace{.1in}}

\newtheorem{theorem}{Theorem}
\newtheorem{problem}{Problem}
\newtheorem{remark}{Remark}
\newtheorem{corollary}{Corollary}

\newtheorem{lemma}{Lemma}
\newtheorem{example}{Example}

\newtheorem{proposition}{Proposition}

\newcommand{\Cov}{\mathrm{cov}}

\usepackage[utf8]{inputenc} 
\usepackage[T1]{fontenc}
\usepackage{url}
\usepackage{ifthen}
\usepackage{cite}

\begin{document}
\title{{Implicit and Explicit Formulas of the Joint RDF for a Tuple of Multivariate Gaussian Sources with Individual {Square-Error} Distortions  }} 


\author{%
  \IEEEauthorblockN{Evagoras Stylianou}
  \IEEEauthorblockA{Chair of Theoretical Information Technology\\
                    Technical University of Munich,             Munich, Germany\\
                    Email: evagoras.stylianou@tum.de\vspace{-3cm}}
  \and
  \IEEEauthorblockN{Charalambos D. Charalambous and Themistoklis Charalambous}
  \IEEEauthorblockA{Department of Electrical and Computer Engineering,\\
                    University of Cyprus,
                    Nicosia, Cyprus\\
                    Email: \{chadcha, charalambous.themistoklis\hspace{-0.005cm}\}@ucy.ac.cy}\vspace{-0.9cm}
}


\maketitle


\begin{abstract} This paper analyzes the joint Rate Distortion Function (RDF) of correlated multivariate Gaussian sources with individual square-error distortions. Leveraging Hotelling's canonical variable form, presented is a closed-form characterization of the joint RDF, that involves {a system of nonlinear equations. Furthermore, for the special case of symmetric distortions (i.e., equal distortions), the joint RDF is explicitly expressed in terms of} two water-filling variables. The results greatly improve our understanding and advance the development of closed-form solutions of the  joint RDF for multivariate Gaussian sources with individual square-error distortions.
\end{abstract}

\section{Introduction}

Shannon's seminal papers  \cite{shannon2001mathematical,shannon1959coding} established the fundamentals of lossy compression theory, outlining the minimum number of bits required per source symbol for compressing a source under a fidelity criterion. This fundamental concept is captured by the rate distortion function (RDF).
The RDF has been computed for various sources, including memoryless Gaussian sources \cite{berger:B1971}, multivariate Gaussian sources \cite{cover}, and sources with memory \cite{wyner1971bounds}. \ES{The extension of the RDF to two correlated sources with individual distortions was pioneered by Gray \cite{gray1973}, who coined the term \emph{joint} RDF}


The joint RDF with individual distortions, depicted in Fig.~\ref{def:JointRDFGauss},  has found applications in various lossy network compression problems \cite{sourcecoding, viswanatha2014lossy, xu2015lossy, feng2006rate,zamir1999multiterminal,nayak2010successive}. This includes its use in analyzing the optimal operational rate region of the Gray-Wyner network \cite{sourcecoding}, and in  defining Wyner's \emph{lossy} common information \cite{viswanatha2014lossy,xu2015lossy}. Moreover, recently it was used to characterize the semantic information of lossy compression problems~\cite{liu2021rate,liu2022indirect,stavrou2022rate,guo2022semantic} \ES{and to characterize the nonanticipative joint RDF \cite{charalambous2021joint}. Therefore, having closed-form solutions for relevant sources, such as Gaussian random variables (RVs), is desirable in the literature.}

For scalar-valued Gaussian RVs with compound square{-}error distortion, closed-form solutions were partially provided in \cite{gray1973,sourcecoding}. The complete solution, i.e., with two individual square{-}error distortion{s} was computed in \cite{xiao}, and used therein to establish a tight lower bound on the sum rate of the multi-terminal inner rate region \cite{Berger1979, Tung, Oohama2005}. This solution was further employed in \cite{lapidoth} (with an alternative derivation) to analyze the power-versus-distortion trade-off for transmitting two scalar-valued Gaussian sources over an additive white Gaussian noise channel. Additionally, in \cite{viswanatha2014lossy, xu2015lossy}, the expression was applied to compute Wyner's lossy common information. The authors \cite{charalambous2022realization} used the expression to compute rates on the Gray-Wyner rate region for Gaussian sources with square-error distortions.
In \cite{stylianou2021joint}, the authors extend the joint RDF to a tuple of multivariate Gaussian RVs with two individual square{-}error distortions, and derived closed-form expressions for a subset of the distortion region. The investigation shed light on the increased complexity inherent in the joint RDF for a tuple of sources with individual distortions compared to the  RDF for single sources and distortions.  

Despite numerous applications in lossy network compression, a complete closed-form expression for a tuple of multivariate Gaussian sources with individual square-error distortions is absent. This paper bridges the gap, deriving expressions based on solutions of non-linear equations using Hotelling's canonical variable form (CVF) \cite{hotelling1936relations} to make the optimization problem tractable. While closely aligned with the water-filling principle, the changes in water level are obscured by the necessity of solving nonlinear equations. For individual symmetric distortions, we demonstrate that the solution to the joint RDF {is} explicitly expressed in terms of  {\it two water-filling variables} rather than a single water-filling variable of the conventional RDF \cite{cover}. \vspace{-0.2cm}

\begin{figure}[t]
  \centering
\includegraphics[height=2.5cm,width=\columnwidth]{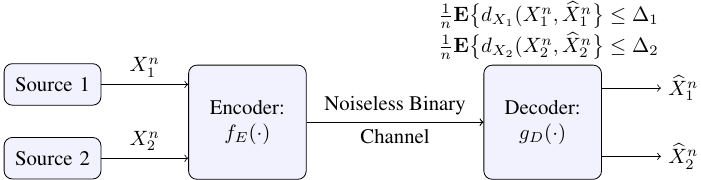}
\caption{ Lossy Compression of correlated sources with individual distortion criteria.}
 \vspace{-0.6cm}
  \label{fig:jointRDF}
\end{figure}




\section{problem statement and preliminaries}
 \subsection{Notation}
Let  ${\mathbb Z}$,  ${\mathbb Z}_+$, $\mathbb{R}$,  be the set of integers, positive integers, {and} real numbers,  respectively,  and ${\mathbb Z}_+^n\tri \{1,2, \ldots,n \}, n \in {\mathbb Z}_+$.
 The expression $\mathbb{R}^{n \times m}$ {for $n,m \in {\mathbb Z}_+$,  denotes the set of $n$ by $m$ matrices with entries 
 the  real numbers.
  For a symmetric matrix $Q \in {\mathbb R}^{n \times n}$, the notation  $Q \succeq 0$ (resp. $Q \succ 0$) means  the matrix is positive semi-definite (resp. positive definite). }
For any  matrix $A\in \mathbb{R}^{p\times m}$, we denote its transpose by $A\T$, its pseudoinverse by $A^\dagger$,  and for $m=p$,  we denote its trace and its determinant  by  $\trace(A)$ and $\det\big(A\big)$, respectively. By $\diag(A)$ we denote the diagonal matrix with entries those of A and zero elsewhere. By  $I_p$ and $0_{t\times s}$ we denote the ${p\times p}$ identity matrix and the ${t\times s}$ matrix with all zero entries respectively. By $X \in G(0, Q_{X})$ we denote a finite-dimensional Gaussian RV with zero mean and covariance matrix $Q_{X}\succeq 0$.
The mutual information between two  RVs $X$ and $Y$ is denoted by $I(X;Y)$ defined as in~\cite{gallager1968information}.

\vspace*{-0.3cm}
\subsection{Problem statement}
Consider a tuple of RVs, $X_i: \Omega \rightarrow \mathbb{X}_i,\:i=1,2$
where $\mathbb{X}_i$ represents arbitrary measurable spaces equipped with corresponding metrics 
and a tuple of  distortion functions,  $d_{X_i} :{\mathbb X}_i \times \widehat{{\mathbb X}}_i \rar [0,\infty),i=1,2$, where  $\widehat{\mathbb X}_i \subseteq {\mathbb X}_i,~i=1,2$. Then, the joint RDF is defined as follows:\vspace{-0.1cm}
\begin{align} 
&{R}_{X_1,X_2}(\Delta_1,\Delta_2)  \tri  \inf_{\pazocal{M}(\Delta_{1},\Delta_{2})} I(X_1,X_2;\widehat{X}_1,\widehat{X}_2) \in [0,\infty], \label{def:JointRDFGauss}\\
&\pazocal{M}(\Delta_{1},\Delta_{2}) \tri\big \{\widehat{X}_1: \Omega \rightarrow  \widehat{{\mathbb X}}_1,\;\;\widehat{X}_2: \Omega \rightarrow  \widehat{{\mathbb X}}_2 \big |\; \mathbf{P}_{X_1,X_2,\widehat{X}_1,\widehat{X}_2} \text{has}\nonumber \\ & \text{$(X_1,X_2)$-marginal ${\bf P}_{X_1,X_2}$},  \; \mathbf{E} \big\{ d_{X_i}(X_i,\widehat{X}_i) \big \} \le \Delta_{i}, \;i=1,2    \big \}. \nonumber 
\end{align}
The joint RDF characterizes the infimum of all achievable  rates of a sequence of  rate distortion codes, $(f_E, g_D)$, as depicted in Fig.~\ref{fig:jointRDF},  of  reconstructing $(X_1^n, X_2^n) \tri  \{(X_{1,t}, X_{2,t})\}_{t=1}^n$ with ${\bf P}_{X_{1,t}, X_{2,t}}={\bf P}_{X_1,X_2}, \forall t$, by  $(\widehat{X}_1^n, \widehat{X}_2^n)$  with  average distortions $\frac{1}{n}{\bf E}\{d_{X_i} (X_i^n, \widehat{X}_i^n)\}\leq \Delta_i, i=1,2$, for sufficiently large~$n$.

We focus on a tuple of jointly independent and identically distributed multivariate Gaussian RVs, $(X_1^n, X_2^n)$, $X_{i,t}:\Omega \rar \mathbb{R}^{p_i},\; X_{i,t} \in G(0, Q_{X_i}),\forall t,i=1,2$,  
 $(X_{1,t}, X_{2,t}) \in G(0,Q_{(X_{1}, X_{2})})$, 
 and individual   square-error distortion  functions, {given by} \vspace{-0.1cm}
\begin{align}
&Q_{(X_{1}, X_{2})} =  {\mathbf E} \bigg\{ \begin{pmatrix} X_{1} \\ X_{2}  \end{pmatrix}  \begin{pmatrix}X_{1} \\ X_{2}  \end{pmatrix}\T \bigg\}=   \begin{pmatrix} Q_{X_1} & Q_{X_1,X_2} \\ Q_{X_1,X_2}\T & Q_{X_2} \end{pmatrix} \label{prob_1}\\
& \widehat{X}_{1,t}: \Omega  \rar  {\mathbb R}^{p_1}, \hso \widehat{X}_{2,t}: \Omega  \rar    {\mathbb R}^{p_2} \hso
\forall t,\label{prob_8} \\
& d_{X_i} (x_i^n, \widehat{x}_i^n)= \frac{1}{n} \sum_{t=1}^n ||x_{i,t}-\widehat{x}_{i,t}||^{{2}},\;\;\;i=1,2.  \label{prob_9}
\end{align}
Then, the joint RDF for (\ref{prob_1})-(\ref{prob_9}) is given by \cite{stylianou2021joint} \vspace{-0.1cm}
\begin{align}
&R_{X_1,X_2}(\Delta_{1},\Delta_{2})
= \inf_{{{\pazocal{Q}}}(\Delta_{1},\Delta_{2})}    \frac{1}{2}\log\bigg(\frac{\det\big( Q_{(X_1,X_2)}\big)}{\det\big(\Sigma_{(E_1,E_2)}\big) } \bigg) ,  \label{eq:JointRDFOpti_new}
\end{align}
where\vspace{-0.1cm}
\begin{align}
& {{\pazocal{Q}}}(\Delta_{1},\Delta_{2}) \tri  \big \{\Sigma_{(E_1,E_2)}\succeq 0   \big |  \: \: Q_{(X_1,X_2)} \succeq \Sigma_{(E_1,E_2)} \succeq 0, \nonumber \\ &\hspace{2.1cm}\trace \big ( \Sigma_{E_1}\big) \le \Delta_{1}, \; \trace \big ( \Sigma_{E_2}\big) \le \Delta_{2} \big \},\nonumber \\
& E \tri \begin{pmatrix}
E_1 \\ E_2
\end{pmatrix}, \quad E_i \tri  X_i - \widehat{X}_i,\;  i=1,2,\nonumber \\&   \Sigma_{(E_1,E_2)} = \Cov (E,E) = \begin{pmatrix}
\Sigma_{E_1} & \Sigma_{E_1,E_2}  \\
\Sigma_{E_1,E_2} ^T & \Sigma_{E_2} \\
\end{pmatrix}.\nonumber 
\end{align}

\subsection{Reformulation of the joint RDF for RVs in the CVF }
In this section we reformulate the joint RDF {given by 
 (\ref{eq:JointRDFOpti_new}), for $(X_1, X_2)$ in CVF}  \cite{hotelling1936relations, van2021control}.  We focus on the algorithm for transforming the tuple of Gaussian RVs $(X_1, X_2)$ into the CVF.

\begin{Algorithm}{ \cite{van2021control}} \label{alg:cvf}Transformation of a covariance matrix $Q_{(X_1,X_2)} \in \mathbb{R}^{(p_1+p_2)\times(p_1+p_2)}$, satisfying $Q_{(X_1,X_2)} \succeq 0$, {$Q_{X_i} \succ 0, i=1,2$} to its CVF. \\
(1) Perform the singular value decompositions (SVDs), \vspace{-0.1cm}
\begin{align*}
&Q_{X_i} = U_iD_iU_i^T,\quad U_iU_i\T = I_{p_i},\quad i=1,2, \\
&D_i = \diag(d_{1,1},\dots,d_{i,p_{i}})\in \mathbb{R}^{p_1\times p_1},\;\;d_{i,1} \ge \dots \ge d_{i,p_i}>0
\end{align*}
(2) Perform the SVD, \vspace{-0.1cm}
\begin{align*}
    D_1^{-\frac{1}{2}}U_1^T Q_{X_1,X_2}U_2D_2^{-\frac{1}{2}} = U_3D_3U_4^T,
\end{align*}
where $U_3$ and $U_4$ are orthogonal and\vspace{-0.1cm}
\renewcommand\arraystretch{0.92}
\begin{align}
&D_3 = \mathrm{Block}\text{-}\diag\big(I_{p_{11}},D_4,0_{p_{13} \times  p_{23}} \big)\in \mathbb{R}^{p_1\times p_2},\nonumber \\&D_4 = \diag(d_{4,1},\dots,d_{4,p_{12}})\in \mathbb{R}^{p_{12}\times p_{22}},\nonumber \\
& p_i = p_{i1} + p_{i2} + p_{i3},\;\;i=1,2,\;\; p_{11} = p_{21},\;\; p_{12} = p_{22},\nonumber
\end{align}
where $1 > d_{4,1} \ge \dots \ge d_{4,p_{12}}>0$ are refer to as canonical correlation coefficients.\\
(3) Compute the new covariance matrix and the linear transformation $(S_1,S_2)$ of $(X_1,X_2)$ as follows:  \renewcommand\arraystretch{1.3}
\begin{align}
Q_{\mathrm{\mathrm{cvf}}} = \begin{pmatrix}
I_{p_1} & D_3 \\
D_3^T & I_{p_2}
\end{pmatrix},\;\; S_1 =U_3^TD_1^{-\frac{1}{2}}U_1^T,\;\;S_2 = U_4^TD_2^{-\frac{1}{2}}U_2^T.\nonumber
\end{align}
Then, $X_i^\mathrm{c}=S_i X_i,\;i=1,2$ are such that $(X_1^\mathrm{c}, X_2^\mathrm{c})\in  G(0, Q_{\mathrm{\mathrm{cvf}}})$.
\end{Algorithm} 
\begin{remark} \label{rem:cvf}
We focus on RVs in  CVF $(X_1,X_2)=(X_1^\mathrm{c}, X_2^\mathrm{c}) \in G(0,Q_{\mathrm{cvf}})$, and, without loss of generality, we consider only correlated components, i.e., with indices $p_{11} = p_{21} = p_{13} = p_{23} = 0$, $p_{12} = p_{22} = n$, specified~by
\begin{align*}
     & Q_{(X_1,X_2)}=Q_{\mathrm{cvf}}
   = \begin{pmatrix}          I_n & D \\
          D & I_n
        \end{pmatrix}
        ,  \;\; D
 = \diag (d_1, d_2, \ldots, d_n) \in  {\mathbb R}^{n \times n}.
\end{align*}
\end{remark}
First, we derive a lower bound on the mutual information.

\begin{corollary} \label{cor_CVFjoint} Consider $(X_1,X_2)\in G(0,Q_{(X_1, X_2)})$ in CVF of Remark~\ref{rem:cvf}. The following statements hold.\\
(a) The mutual information satisfies the lower bound,
\begin{align}
I(X_1, X_2; \widehat{X}_1, \widehat{X}_2) \geq& \sum_{i=1}^n I(X_{1,i}, X_{2,i} ; \widehat{X}_{1,i} \widehat{X}_{2,i} ). \nonumber
\end{align}
(b) The lower bound is achieved if the RVs $(X_{1,i}, X_{2,i}, \widehat{X}_{1,i}, \widehat{X}_{2,i})$ are independent of $(X_{1,j}, X_{2,j}, \widehat{X}_{1,j}, \widehat{X}_{2,j}),\;\forall i \neq j$. This implies that the error RVs $(E_{1,i},E_{2,j})$ are mutually independent for all $i\neq j$ and the optimal test channel of the joint RDF is induced by the diagonal matrices,
\begin{align}
 \Sigma_{E_j}=\diag(\Delta_{j,1},\dots,\Delta_{j,n}), \;j=1,2, \;\; \Sigma_{E_1,E_2} = \diag(\widehat{d}_1,\dots,\widehat{d}_n).\nonumber
 \end{align}
\end{corollary}
\begin{proof} It follows by well-known properties of mutual information for general RVs. 
\end{proof}
\begin{lemma}(Extended Schur's Lemma \cite{albert1969conditions,rami2000linear})
\label{lem_schur}\\
 Let $M =M\T$, $R =R\T$ and $N$ be given matrices of appropriate dimensions. {The{n,} statements (a) and (b)} below are equivalent.\renewcommand\arraystretch{0.7} 
 \begin{align}
 &(a)\; \begin{pmatrix}
 M & N \\ N\T & R
 \end{pmatrix}\succeq 0, \nonumber \\ &(b) \;M-N R^\dagger N\T \succeq 0,\;\;\;\; R \succeq 0,\;\;\;\; { \big(I-R R^\dagger\big)N\T =0.} \nonumber
 \end{align}
\end{lemma} 
We can replace $Q_{(X_1,X_2)}-\Sigma_{(E_1,E_2)}\succeq 0$ by the expression in Lemma~\ref{lem_schur} item (b) and use $\det(\Sigma_{(E_1,E_2})=\det(\Sigma_{E_1})\det(\Sigma_{E_2}-\Sigma_{E_1,E_2}\Sigma_{E_1}\Sigma_{E_1,E_2}\T)$. This allows the following reformulation,\vspace{-0.1cm}
\begin{problem} 
\label{prob-main}
The joint RDF given in  \eqref{eq:JointRDFOpti_new} for  a tuple of RVs {$(X_1,X_2)\in G(0,Q_{(X_1,X_2)})$ in} CVF specified in Remark~\ref{rem:cvf}, is re-formulated as follows: \vspace{-0.1cm}
 \begin{align*}
 &R_{X_1,X_2}(\Delta_1,\Delta_2) = \nonumber\\
 &\inf_{\Delta_{1,i},\Delta_{2,i}\ge 0,\; \widehat{d}_i,\in [0,1), \:\forall i \in {\mathbb Z}_+^n} \frac{1}{2} \sum_{i=1}^n \log \Bigg ( \frac{(1 -d_i^2)}{(\Delta_{1,i}\Delta_{2,i}-\widehat{d}_i^2)}  \Bigg ), \\
&\mathrm{s.t.}\;\;\big(1 - \Delta_{1,i}\big)-\big(d_i -\widehat{d}_i \big )^2\big(1 - \Delta_{2,i}\big)^{\dagger} \ge 0,\;\forall i \in \mathbb{Z}_+^n,\\
&\hspace{0.6cm} \big(d_i -  \widehat{d}_i \big )\big ( 1 - \big(1 - \Delta_{2,i}\big) \big(1 - \Delta_{2,i}\big)^{\dagger} \big)=0, \;\forall i \in \mathbb{Z}_+^n,\\
&\hspace{0.6cm}(1 - \Delta_{2,i})\ge 0,\;\;\forall i\in\mathbb{Z}_+^n,\;\; \sum_{i=1}^n\Delta_{1,i} \le \Delta_{1},\;\;  \sum_{i=1}^n\Delta_{2,i} \le \Delta_{2}.  
\end{align*}
\end{problem}

The transformation of the Gaussian RVs $(X_1, X_2)$ into their CVF renders the optimization problem for the joint RDF more tractable. This is evident because the transformation allows us to optimize with respect to the diagonal elements of the blocks $\Sigma_{E_1}$, $\Sigma_{E_2}$, and $\Sigma_{E_1,E_2}$.\vspace{-0.2cm}

\section{{implicit and explicit expressions of the joint rdf}}\vspace{-0.2cm}
Here, we address Problem~\ref{prob-main} for {arbitrary $(\Delta_1, \Delta_2)$  and for the special case of symmetric distortion, i.e., $\Delta_1=\Delta_2$. 
\subsection{Implicit expressions for general distortions}
Consider Problem~\ref{prob-main} for any distortion pair $(\Delta_1,\Delta_2)$.\vspace{-0.1cm}
\begin{theorem}  \label{thm:CovJoint}
Consider Problem~\ref{prob-main} and the sets of parameters $\{\Delta_{1,i},\Delta_{2,i}\}_{i=1}^n$, $\{{d}_i\}_{i=1}^n$ and  $\{\widehat{d}_i\}_{i=1}^n$. Define the following distortion regions for 
$i,j \in \{0,1\},\;i\neq j$,\vspace{-0.1cm}
\begin{align*}
{\mathcal{M}}_{i} = \Big \{  (\Delta_{1},\Delta_{2}) \in \mathbb{R}_+^2\Big | \;\; \Delta_{i} > n - \sum_{i=1}^n d_i^2 \Big (1-\min\Big \{1, \frac{\Delta_j}{n},\Big \}\Big) \Big \}.
\end{align*}
Let  $\kappa$ be  the number of non-zero elements of the set $\{\widehat{d}_i\}_{i=1}^n$ and  $\ell$ be the number of quadruples  $(\Delta_{1,i}, \Delta_{2,i}, d_i, \widehat{d}_i)$  such that with $\Delta_{1,i}=\Delta_{2,i} = 1$ and $\widehat{d}_i = d_i$. Assume $(\kappa, \ell)$  is  known.\\ Then, the correlation coefficients $\{\widehat{d}_i\}_{i=1}^n$ are given by\vspace{-0.1cm}
\begin{align}
\widehat{d}_i = \begin{cases} 
d_{i} - \sqrt{(1-\Delta_{1,i})(1-\Delta_{2,i})}, 
& \forall i\in \mathbb{Z}_+^\kappa, \\
0, & \forall i \in {\mathbb Z}_+^n \setminus {\mathbb Z}_+^\kappa
\end{cases} \label{coref},
\end{align}
and $\{\Delta_{1,i},\Delta_{2,i}\}_{i=1}^n$ are obtained as follows.\\
Case (a): If $(\Delta_1,\Delta_2) \notin \mathcal{M}_k $ for $k=1,2$ and $\kappa = 0$, then, \vspace{-0.1cm}
\begin{align}
    \Delta_{1,i} = \frac{\Delta_1}{n}, \quad \Delta_{2,i} = \frac{\Delta_2}{n},\quad  \forall i \in \mathbb{Z}_+^n.
\end{align}
Case (b): If $(\Delta_1,\Delta_2) \notin \mathcal{M}_k$ for $k=1,2$ and $\kappa< n$, then $\{\Delta_{1,i},\Delta_{2,i}\}_{i=1}^n$, are determined for all $ i \in \mathbb{Z}_+^\kappa$ by \vspace{-0.1cm}
\begin{align}
& \frac{1- d_{i}\sqrt{\frac{1-\Delta_{s,i}}{1-\Delta_{t,i}}}}{\Delta_{s,i}\Delta_{t,i} -\big ( d_{i} - \sqrt{(1-\Delta_{s,i})(1-\Delta_{t,i})}\big )^2} + \frac{n-\kappa}{\Delta_{t} - \sum_{i=1}^\kappa\Delta_{t,i}} = 0, \label{eqsys11}
\end{align}
where $s,t \in \{1,2\},\;s\neq t$ and for all  $i \in \mathbb{Z}_+^n\setminus \mathbb{Z}_+^\kappa$\vspace{-0.1cm}
\begin{align}
&\Delta_{1,i} - \frac{\Delta_{1} - \sum_{i=1}^\kappa\Delta_{1,i}}{n-\kappa} = 0,\;\;\;\;\Delta_{2,i} - \frac{\Delta_{2} - \sum_{i=1}^\kappa\Delta_{2,i}}{n-\kappa} = 0.  \label{eqsys41}
\end{align}
Case (c):  If $(\Delta_1,\Delta_2) \notin \mathcal{M}_k$ for $k=1,2$, $\kappa=n$ and $\ell \ge 0$, then $\{\Delta_{1,i},\Delta_{2,i}\}_{i=1}^{n-\ell}$, are determined for all $ i \in \mathbb{Z}_+^{n-\ell}$ by
 \begin{align}
 & -\frac{1- d_{i}\sqrt{\frac{1-\Delta_{s,i}}{1-\Delta_{t,i}}}}{2(\Delta_{s,i}\Delta_{t,i} -\big ( d_{i} - \sqrt{(1-\Delta_{s,i})(1-\Delta_{t,i})}\big )^2)} + \lambda_t = 0,\label{eq:nonsysteml1}\\ & \sum_{i =1}^{n-\ell} \Delta_{1,i} = \Delta_1-\ell,\quad \sum_{i =1}^{n-\ell} \Delta_{2,i} = \Delta_2-\ell,\label{eq:nonsysteml2}
\end{align}
where $s,t \in \{0,1\},\;s\neq t$ and $\{\Delta_{1,i} = \Delta_{2,i} = 1\}_{i=n-\ell}^n$.\\
Case (d):  If $(\Delta_1,\Delta_2) \in \mathcal{M}_1$, then  $\kappa=n$ and $\{\Delta_{1,i},\Delta_{2,i}\}_{i=1}^n$, are determine, for all $i\in \mathbb{Z}_+^{n}$ by
\begin{align}
& \hspace{-0.15cm}\Delta_{2,i} =\min\Big \{1,\frac{\Delta_2}{n}\Big\}, \;\;  \Delta_{1,i} = 1-d_i^2(1-\Delta_{2,i}),\;\; \widehat{d}_i = d_i\Delta_{2,i}.\nonumber
\end{align}
Case (e): If $(\Delta_1,\Delta_2) \in \mathcal{M}_2$, then  $\kappa=n$ and $\{\Delta_{1,i},\Delta_{2,i}\}_{i=1}^n$, are determine, for all $i\in \mathbb{Z}_+^{n}$ by 
\begin{align}
& \hspace{-0.15cm}\Delta_{1,i} = \min\Big \{1,\frac{\Delta_1}{n}\Big \}, \;\; \Delta_{2,i} = 1-d_i^2(1-\Delta_{1,i}),\;\; \widehat{d}_i = d_i\Delta_{1,i}. \nonumber
\end{align}
\end{theorem}
\begin{proof} \ES{By application of the KKT conditions on Prob.~\ref{prob-main}}.\vspace{-0.4cm}
\end{proof}
It is evident that obtaining a closed-form solution to the non-linear equations of Theorem~\ref{thm:CovJoint} may not be possible, indicating the complexity of the optimization problem even when $(X_1,X_2)$ are in CVF. This challenges the pursuit of a closed-form solution when $(X_1,X_2)$ is not in CVF, which involves matrices $\Sigma_{E_{i}}, i=1,2$ and $\Sigma_{E_1,E_2}$ that are not~diagonal.
\ES{In general, consider the following remark.
\begin{remark}
    The problem of calculating the RDF and optimal encodings is challenging  because no closed-form solution for optimal encodings is known. Therefore, the class of Blahut-Arimoto algorithms was developed for the numerical calculation of the relevant quantities \cite{blahut1972computation}. However, in \cite{boche2023algorithmic} it was shown that for Blahut-Arimoto algorithms for the calculation of Shannon capacity, the optimal input distribution is not Turing computable. In \cite{lee2023computability} the impossibility of Turing computability of optimisers for a larger class of information-theoretic tasks was shown. It is an interesting open question to show this impossibility of Turing computability of optimal encoders for the RDF as well.
\end{remark}}

Subsequently, we provide the expressions of the joint RDF. 
 
\begin{theorem}\label{thm:jointRDFGeneralCVF}
Consider the statements of Theorem \ref{thm:CovJoint} and assume that $( \kappa,\ell)$ are known. The joint RDF is given by, \\
Case (a):  If $(\Delta_1,\Delta_2) \notin \mathcal{M}_k$ for $k=1,2$ and $\kappa = 0$, then for $(\Delta_1,\Delta_2) \in {\mathcal{D}}_{0}$, \vspace{-0.1cm}
\begin{align}
&R_{X_1,X_2}(\Delta_1,\Delta_2)   =  \frac{1}{2}\sum_{i=1}^n \log \bigg ( \frac{n^2(1-d_i^2)}{\Delta_{1}\Delta_{2} }\bigg ), \nonumber
\end{align}
where\vspace{-0.1cm}
\begin{align}
{\mathcal{D}}_{0}=   \Big \{&  (\Delta_{1},\Delta_{2}) \in \mathbb{R}_+^2\Big |\;  d_1^2 \le (1-\frac{\Delta_1}{n})(1-\frac{\Delta_2}{n}), \Delta_{1}, \Delta_{2}<n \Big \}. \nonumber
\end{align}
Case (b): If $(\Delta_1,\Delta_2) \notin \mathcal{M}_k$ for $k=1,2$ and $\kappa < n$, then for $(\Delta_1,\Delta_2) \in {\mathcal{D}}_{\kappa}$,\vspace{-0.1cm}
\begin{align}
&R_{X_1,X_2}(\Delta_1,\Delta_2)   =  \frac{1}{2}\sum_{i=\kappa+1}^n \log \bigg ( \frac{1-d_i^2}{\Delta_{1,i}\Delta_{2,i} }\bigg ) \nonumber \\ &+\frac{1}{2} \sum_{i=1}^\kappa \log \bigg ( \frac{(1-d_i^2)}{\Delta_{1,i}\Delta_{2,i} -\big ( d_{i} - \sqrt{(1-\Delta_{1,i})(1-\Delta_{2,i})}\big )^2 }\bigg ) ,\nonumber 
\end{align}
where\vspace{-0.1cm}
\begin{align}
{\mathcal{D}}_{\kappa}=   \Big \{&  (\Delta_{1},\Delta_{2}) \in \mathbb{R}_+^2\Big |\;  d_j^2 \le (1-\Delta_{1,j})(1-\Delta_{2,j}),\;\forall i \in \mathbb{Z}_+^{n}\setminus\mathbb{Z}_+^{\kappa},\nonumber\\ &  (1-\Delta_{1,i})(1-\Delta_{2,i}) \le d_{i}^2 \le \min \Big \{ \frac{(1-\Delta_{1,i})}{(1-\Delta_{2,i})},\frac{(1-\Delta_{2,i})}{(1-\Delta_{1,i} )}\Big \},\nonumber \\ &\;\forall i \in \mathbb{Z}_+^{\kappa},\;\; \Delta_{1,i}<1,\; \Delta_{2,i}<1,\;\forall i\in \mathbb{Z}_+^{n}\Big \}, \nonumber
\end{align}
where $\{\Delta_{1,i},\Delta_{2,i}\}_{i=1}^n$ are obtained by the solution of the non-linear equations \eqref{eqsys11} and \eqref{eqsys41}.\\
Case (c): If $(\Delta_1,\Delta_2) \notin \mathcal{M}_k$ for $k=1,2$, $\kappa = n$ and $\ell\ge 0$, then for $(\Delta_1,\Delta_2) \in~\widehat{\mathcal{D}}_\ell$, \vspace{-0.1cm}
\begin{align}
R_{X_1,X_2}&(\Delta_1,\Delta_2)   =  \nonumber \\  & \hspace{-0.5cm}\frac{1}{2} \sum_{i=1}^{n-\ell} \log \bigg ( \frac{(1-d_i^2)}{\Delta_{1,i}\Delta_{2,i} -\big ( d_{i} - \sqrt{(1-\Delta_{1,i})(1-\Delta_{2,i})}\big )^2 }\bigg ),\nonumber 
\end{align}
where \vspace{-0.1cm}
\begin{align}
\widehat{\mathcal{D}}_\ell  = \big  \{ & (\Delta_{1},\Delta_{2}) \in \mathbb{R}_+^2 \Big |\;d_{i}^2 \ge (1-\Delta_{1,i})(1-\Delta_{2,i}),\;\Delta_{1,i},\Delta_{2,i}<1,\nonumber\\ &d_{i}^2 
 \le \min \Big \{ \frac{(1-\Delta_{1,i})}{(1-\Delta_{2,i})},\frac{(1-\Delta_{2,i})}{(1-\Delta_{1,i} )}\Big \}, \;\forall i\in  \mathbb{Z}_+^{n-\ell}\big \},\nonumber
\end{align}
where $\{\Delta_{1,i},\Delta_{2,i}\}_{i=1}^n$ are obtained by the solution of the non-linear equations \eqref{eq:nonsysteml1} and \eqref{eq:nonsysteml2}.\\
Case (d): If $(\Delta_1,\Delta_2) \in \mathcal{M}_1$, then, $\Delta_{2,i} = \min \big \{1,\frac{\Delta_2}{n} \big\}$ and, \vspace{-0.1cm}
\begin{align}
R_{X_1,X_2}(\Delta_{1},\Delta_{2})  = R_{X_2}(\Delta_{2}) = \frac{1}{2} \sum_{i=1}^n\log \bigg ( \frac{1}{\Delta_{2,i}} \bigg ).\;\;\nonumber
\end{align}
Case (e): If $(\Delta_1,\Delta_2) \in \mathcal{M}_2$, then, $\Delta_{1,i} = \min \big \{1,\frac{\Delta_1}{n} \big\}$, and, \vspace{-0.1cm}
\begin{align}
R_{X_1,X_2}(\Delta_{1},\Delta_{2})= R_{X_1}(\Delta_{1}) = \frac{1}{2} \sum_{i=1}^n\log \bigg( \frac{1}{\Delta_{1,i}} \bigg ) \nonumber.
\end{align}
\end{theorem}
\begin{proof} Follows from Theorem \ref{thm:CovJoint}.  \vspace{-0.5cm}
\end{proof}

It is evident from Theorem~\ref{thm:CovJoint} and Theorem~\ref{thm:jointRDFGeneralCVF} that the joint RDF follows a distinct water-filling process on both the distortions $\{\Delta_{1,i},\Delta_{2,i}\}_{i=1}^n$ and the correlation coefficients $\{\widehat{d}_i\}_{i=1}^n$. The calculation of the joint RDF initiates by evenly distributing the distortion among all components and setting $\kappa = 0$ (Case (a)). If this allocation is infeasible based on the corresponding region in Theorem~\ref{thm:jointRDFGeneralCVF}, we introduce $\kappa = 1$ correlation coefficient using \eqref{coref}. Additionally, we include two non-linear equations as per \eqref{eqsys11}, and recalculate the distortions (Case (b)). This iterative process continues until a feasible solution is found or $\kappa = n$. If $\kappa = n$ and a feasible solution remains elusive, then we set $\ell = 1$, $\Delta_{1,n} = \Delta_{2,n} = 1$, and $\widehat{d}_n = d_n$ (Case (c)). Subsequently, we eliminate the non-linear equations for $i=n$, solving \eqref{eq:nonsysteml1} and \eqref{eq:nonsysteml1}. The process is reiterated by incrementing $\ell$ and checking the feasibility of the solution. This iterative procedure reveals the presence of two hidden water-filling thresholds embedded within the non-linear equations. 

Moreover, it can be seen from Cases (d) and (e) that, when the difference between the distortions $\Delta_1$ and $\Delta_2$ is sufficiently large, the joint RDF simplifies to the Shannon RDFs. The extent of this difference is determined by the correlation coefficients $\{d_i\}_{i=1}^n$. Consequently, due to the correlation between the two RVs, it becomes feasible to compress only one of the RVs and utilize its reconstruction to estimate the other. This approach yields a reduced distortion for the other~RV. 

Next, we show that {for $n=1$,   the joint RDF in Theorem~\ref{thm:jointRDFGeneralCVF} reproduces the known joint RDF in \cite{xiao}}.

\begin{proposition}
 The joint RDF of Theorem~\ref{thm:jointRDFGeneralCVF} for $n=1$ is equal to the joint RDF of \cite{xiao} given by\vspace{-0.1cm}
 \begin{align}
&R_{X_1,X_2}(\Delta_1,\Delta_2)=\label{jrdf_xl}
\\& \begin{dcases}
    \frac{1}{2} \log^+ \Big ( \frac{1-d^2}{\Delta_1\Delta_2}\Big ) & (\Delta_1,\Delta_2)\in \mathcal{D}^*_1\\
\frac{1}{2} \log^+ \Big (\frac{1-d^2}{\Delta_1\Delta_2 -  ( d - \sqrt{(1-\Delta_1)(1-\Delta_2)} )^2}\Big ) &  (\Delta_1,\Delta_2)\in \mathcal{D}^*_2\\
\frac{1}{2} \log^+ \Big ( \frac{1}{\min \{ \Delta_1, \Delta_2\}}\Big ) & (\Delta_1,\Delta_2)\in \mathcal{D}^*_3
\end{dcases},
\nonumber  
\end{align}
where $\log^+(x) \tri \max(0,\log(x))$ and \vspace{-0.1cm}
\begin{align*}
    &\mathcal{D}^*_1 = \{(\Delta_1,\Delta_2)\in \mathbb{R}_+^2|\;d^2 \le (1-\Delta_1)(1-\Delta_2)\}\\
    &\mathcal{D}^*_3 = \Big \{(\Delta_1,\Delta_2)\in \mathbb{R}_+^2|\;d^2 \ge \min \Big \{ \frac{1-\Delta_2}{1-\Delta_1},\frac{1-\Delta_1}{1-\Delta_2} \Big\}\Big \},\\  & \mathcal{D}^*_2 =(\mathcal{D}^*_1)^\mathrm{c}\cup (\mathcal{D}^*_2)^\mathrm{c}.
\end{align*}
\end{proposition}
\begin{proof} Suppose $(\Delta_1,\Delta_2)\in \mathcal{M}_2$ (Thm.~\ref{thm:jointRDFGeneralCVF}, Case (e)) and suppose $\Delta_2 \ge \Delta_1$. It follows from $\mathcal{M}_2$ that $d^2 > \frac{1-\Delta_2}{1-\Delta_1}$. Next, suppose $(\Delta_1,\Delta_2)\in \mathcal{M}_1$ (Thm.~\ref{thm:jointRDFGeneralCVF}, Case (d)) and suppose $\Delta_1 \ge \Delta_2$, similarly, $d^2  > \frac{1-\Delta_1}{1-\Delta_2}$. Consequently, by using the expression of the joint RDF in Case (e) and (d) of Thm.~\ref{thm:jointRDFGeneralCVF} its obvious that they correspond to the third region in \eqref{jrdf_xl}. Next, suppose that $(\Delta_1,\Delta_2)\in \mathcal{D}_0$ (Thm.~\ref{thm:jointRDFGeneralCVF}, Case (a)), then it follows that $d^2<(1-\Delta_1)(1-\Delta_2)$. In this case, the joint RDF is equal to the first part of~\eqref{jrdf_xl}. Finally, suppose that $(\Delta_1,\Delta_2)\in \mathcal{D}_1$ (Thm.~\ref{thm:jointRDFGeneralCVF}, Case (b)), then its obvious that for this region the joint RDF is equal to the second part of (\ref{jrdf_xl}). \vspace{-0.1cm}
\end{proof} 

Furthermore, we demonstrate the equivalence of the solutions obtained using Theorem~\ref{thm:jointRDFGeneralCVF} and the semi-definite programming (SDP) formulation of the joint RDF using CVX \cite{cvx}. 
Additionally, we illustrate the water-filling processes.

\begin{example}
Consider $X_i:\Omega\rightarrow\mathbb{R}^2,i=1,2$ with a covariance matrix $Q_{\mathrm{cvf}}\in\mathbb{R}^{4\times 4}$ with $D = \diag (0.588,0.271)$. To solve Problem~\ref{prob-main}, we use the SDP formulation of \cite{stylianou2021joint} or iteratively solving the non-linear equations in Theorem~\ref{thm:CovJoint}. \\
Case 1:  For distortions $(\Delta_1,\Delta_2) = (0.3,0.2)$, \vspace{-0.1cm}
\begin{align*}
    \Sigma_{(E_1,E_2)} = \diag(0.15,0.15,0.1,0.1).
\end{align*}
Given that $d_1^2 \le (1-\Delta_1/2)(1-\Delta_2/2)$, the distortions are evenly distributed along the diagonal elements of both $\Sigma_{E_1}$ and $\Sigma_{E_2}$. This precisely corresponds to Case (a) of Theorem~\ref{thm:CovJoint}.\\
Case 2:   For distortions $(\Delta_1,\Delta_2) = (1.3,1.2)$,
\begin{align*}
    \Sigma_{(E_1,E_2)}  =   \begin{pmatrix}
     0.5692  &  0 &\vline&   0.1414 &  0 \cr
    0 &  0.7308 &\vline&   0 &  0 \cr
    \hline
    0.1414 &  0 &\vline&   0.5381 &  0 \cr
    0 &  0 & \vline&  0 &  0.6619
  \end{pmatrix}.
\end{align*}
Here, we investigate Case (b) of Theorem~\ref{thm:jointRDFGeneralCVF} with $\kappa=1$. The introduction of a correlation coefficient ${\widehat{d}_1} = 0.1414$ is a direct consequence of the increased distortion. Further distortion increments will lead to the introduction of additional correlation coefficients. \\
Case 3: For distortions $(\Delta_1,\Delta_2) = (0.5,1.7)$,\vspace{-0.1cm}
\begin{align*}
    \Sigma_{(E_1,E_2)}  =   \begin{pmatrix}
    0.2500 &  0 &\vline&   0.1469 &  0 \cr
    0 &  0.2500 &\vline&   0 &  0.0677 \cr
    \hline
    0.1469 &  0 &\vline&   0.7411 &  0 \cr
    0 &  0.0677 & \vline&  0 &  0.9449
  \end{pmatrix}.
\end{align*}
In this case, $(\Delta_1,\Delta_2)\in \mathcal{M}_1$ hence the joint RDF degenerates to $R_{X_1}(\Delta_1)$. We see that $\Delta_1$ is equally divided along $\Sigma_{E_1}$ and $\Delta_{2,i},\widehat{d}_i,\:i=1,2$ are given by Case (d) of Theorem~\ref{thm:CovJoint}.
\end{example}

\subsection{Water-filling solution for  symmetric distortions
}

{In this section we  apply Theorem~\ref{thm:CovJoint} and Theorem~\ref{thm:jointRDFGeneralCVF} to  the special case of symmetric distortion pairs, i.e., $\Delta_1=\Delta_2 = \Delta$. }
   \begin{figure}[t]
  \centering
    \includegraphics[height = 5.5cm, width=0.99\columnwidth]{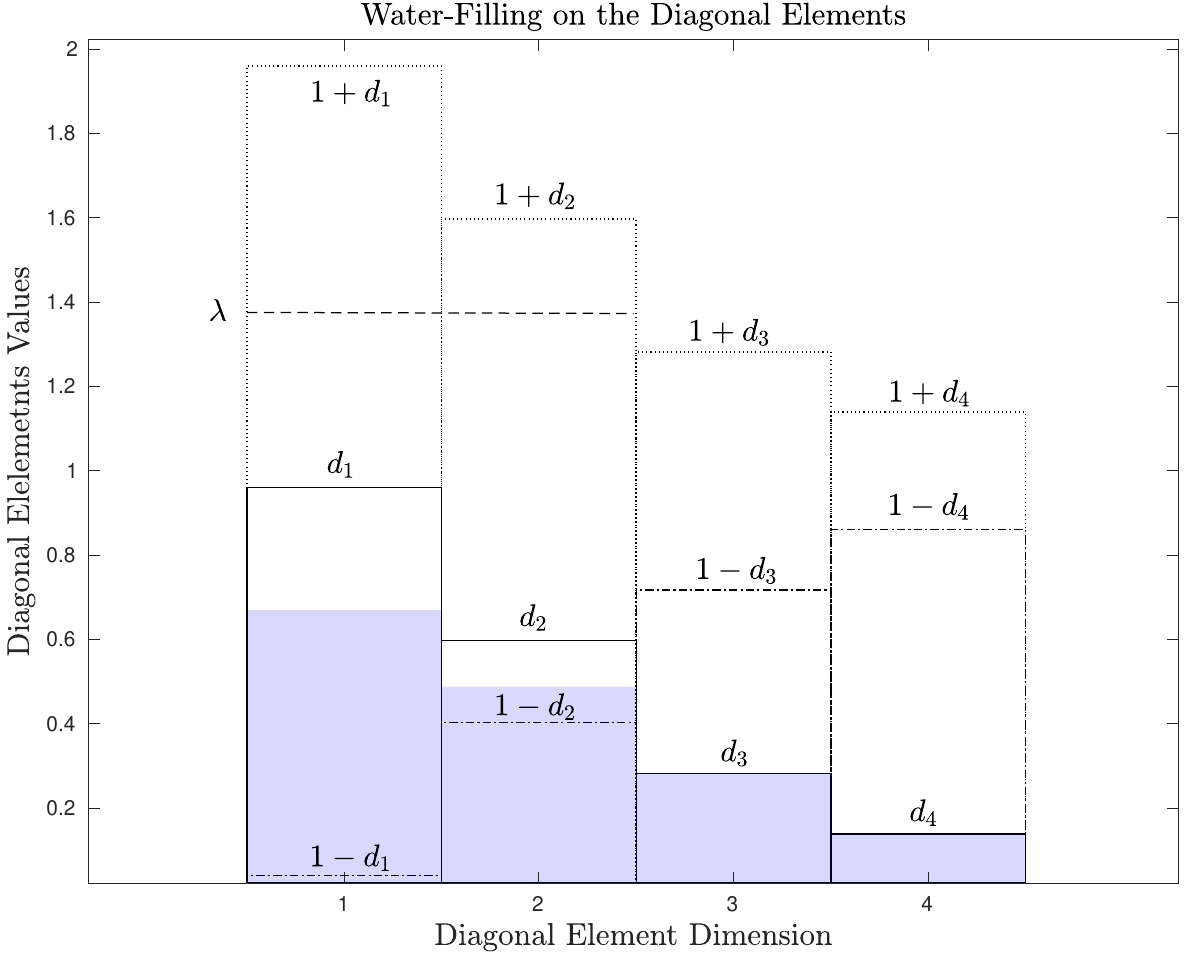}\vspace{-0.4cm}
  \label{fig:condRDF}
\caption{Water-filling solution of the correlation coefficients $\widehat{d}_i,i=1,\dots,4$ for $\Delta = 3.6$ of Example~2.}\label{fig:water1}\vspace{-0.6cm}
\end{figure}
\begin{theorem}\label{cor:jointsym}
{Consider Problem~\ref{prob-main} for symmetric distortions $\Delta_1 = \Delta_2 = \Delta$}. The joint RDF $R_{X_1,X_2}(\Delta)$ is given by\vspace{-0.1cm}
\begin{align}
 R_{X_1,X_2}(\Delta)  & =  \frac{1}{2} \sum_{i=1}^n \log \Bigg ( \frac{1-d_i^2}{\delta_{i}^2 - \widehat{d}_i^2 }\Bigg ), \nonumber
\end{align}
where\vspace{-0.1cm}
\begin{align}
&\delta_i = \begin{dcases} 
\lambda^\prime, & \text{if $\lambda^\prime < 1-d_i$}\\
\frac{\lambda^\prime + (1-d_i)}{2}, &  \text{if $1-d_i \le \lambda^\prime < 1+d_i$}\\
1 & \text{otherwise.}
\end{dcases},\label{eq:jointwater1}\\
&\widehat{d}_i = \begin{dcases} 
0, & \text{if $\lambda^\prime < 1-d_i$}\\
\frac{\lambda^\prime - (1-d_i)}{2}, &  \text{if $1-d_i \le \lambda^\prime < 1+d_i$}\\
d_i & \text{otherwise.}
\end{dcases}, \label{eq:jointwater2}
\end{align}
and $\lambda^\prime$ is computed from  $\sum_{i=1}^n\delta_{i} = \Delta$.
\end{theorem}
\begin{proof} \ES{Follows from Theorem~\ref{thm:CovJoint} for $\Delta_1 =\Delta_2$}.\vspace{-0.1cm}
\end{proof}
The water-filling process on the distortions and the correlation coefficients, as discussed in the previous section, becomes apparent when symmetric distortions are considered. 
In the first step, as described in \eqref{eq:jointwater1} and \eqref{eq:jointwater2}, a non-zero coefficient $\widehat{d}_1$ and the corresponding distortion $\delta_1$ are introduced, since $1-d_1$ is the smallest among $1-d_i$ for $i=1, \ldots, n$. This sequential process continues until reaching $\widehat{d}_n$ and $\delta_n$. In the next step, transitioning from the second to the third cases in \eqref{eq:jointwater1} and \eqref{eq:jointwater2}, a more conventional water-filling occurs. Beyond a certain distortion value, the elements $\delta_i$ and $\widehat{d}_i$ are set to their maximum values $1$ and $d_i$ respectively. In this case, the process starts from $\widehat{d}_n$ and $\delta_n$, as $1+d_i$ is the smallest for $i=n$, and sequentially progresses down to $\widehat{d}_1$ and $\delta_1$.

\begin{figure}[!t]
\centering
  \centering
  \includegraphics[height = 5.5cm, width=0.99\columnwidth]{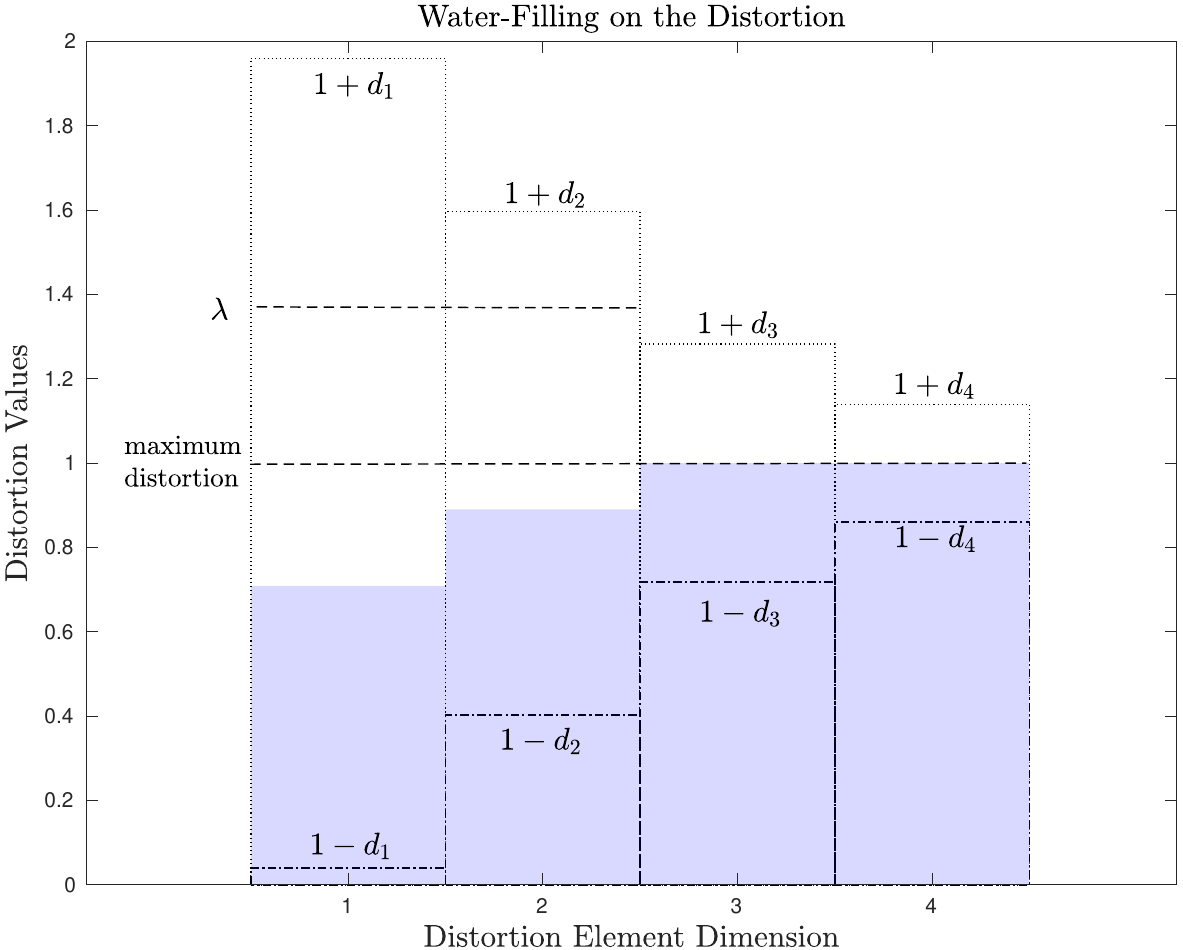}
\caption{Water-filling solution of the correlation coefficients $\delta_i,i=1,\dots,4$ for $\Delta = 3.6$ of Example~2.}  \label{fig:water2} \vspace{-0.6cm}
\end{figure}
\begin{example}
    Consider a tuple of Gaussian RVs $X_i: \Omega \rightarrow {\mathbb R}^4, i=1,2$ with covariance in the CVF with,\vspace{-0.1cm}
    \begin{align*}
        D = \diag \big(0.96, 0.78, 0.40,0.14\big).
    \end{align*}
Let $\Delta = 3.6$, and the outcomes of this process are depicted in Figure~\ref{fig:water1} and \ref{fig:water2}. The figures reveal that for all $i \in \mathbb{Z}_+^4$, $\lambda' > (1-d_i)$, indicating that $\widehat{d}_i > 0$ for all {$i \in \mathbb{Z}_+^{4}$}. For $i=3,4$, it is evident that $\lambda > 1+d_i$, leading to $\widehat{d}_i = d_i$, and $\delta_3 = \delta_4 = 1$. Furthermore, as $\lambda < 1+d_i$ for $i=1,2$, we have $\delta_i$ and $\widehat{d}_i$ according to the second case of  \eqref{eq:jointwater1} and \eqref{eq:jointwater2}. Note that the water-filling procedure starts from index $4$ and progresses to~$1$.\vspace{-0.2cm}
\end{example}


\section{Conclusion}

Our examination of {the optimization problem of  joint RDF for Gaussian RVs in CVF and individual square{-}error distortion functions resulted in nonlinear equations}. Furthermore, we obtained closed-form solutions akin to the well-known water-filling solution for vector Gaussian sources, particularly in scenarios involving symmetric joint RDF. This derived expression significantly contributes to advancing a comprehensive solution for joint RDF, especially for Gaussian RVs not conforming to the CVF.
\vspace{-0.1cm}
\section{Acknowledgments}
\ES{The work of E. Stylianou were supported in part by
the German Federal Ministry of Education and Research (BMBF) within the National Initiative on 6G Communication Systems through the Research Hub 6G-life under Grant 16KISK002, by the BMBF QC-CamNetz Project under Grant 16KISQ077, and by the BMBF Quiet Project under Grant 16KISQ093. The work of T. Charalambous is partly funded by MINERVA from the European Research Council (ERC) under the European Union's Horizon 2022 research and innovation programme (Grant Agreement No. 101044629).}

\balance
\bibliographystyle{IEEEtran}
\bibliography{bibliography}

\end{document}